\documentstyle[extras,11pt]{article}

\newcommand{\al}{\mbox{$\alpha$}}

\newcommand{\minimize}{\mbox{\rm minimize }}
\newcommand{\maximize}{\mbox{\rm maximize }}

\newcommand{\st}{\mbox{\rm subject to }}

\newcommand{\Z}{\mbox{\rm\bf Z}}
\newcommand{\Zplus}{\Z^+}

\newcommand{\implies}{\mbox{${\Rightarrow}$}}

\makeatletter
\def\fnum@figure{{\bf Figure \thefigure}}
\def\fnum@table{{\bf Table \thetable}}

\long\def\@mycaption#1[#2]#3{\addcontentsline{\csname
 ext@#1\endcsname}{#1}{\protect\numberline{\csname
  the#1\endcsname}{\ignorespaces #2}}\par
     \begingroup
       \@parboxrestore
          \small
       \@makecaption{\csname fnum@#1\endcsname}{\ignorespaces
#3\endgroup}
      }
\begin{document}

\title{A Market for Air Traffic Flow Management}

\author{{\Large\em Vijay V. Vazirani}\thanks{College of Computing, Georgia Institute of Technology, Atlanta, GA 30332--0280.
Email: {\sf vazirani@cc.gatech.edu}.
Research supported by NSF Grants CCF-0728640 and CCF-0914732, and a Google Research Grant.}
}

%\setcounter{page}{0}

%\date{}
%\date{September 1, 2011}
\maketitle

\begin{abstract}
The two somewhat conflicting requirements of efficiency and fairness make ATFM an unsatisfactorily solved problem, despite
its overwhelming importance. In this paper, we present an economics motivated solution that is based on the notion of a free market. 
Our contention is that in fact the airlines themselves 
are the best judge of how to achieve efficiency and our market-based solution gives them the ability to pay, at the going rate, to buy 
away the desired amount of delay on a per flight basis.

The issue of fairness is simply finessed away by our solution -- whoever pays gets smaller delays. We show how
our solution has the potential of enabling travelers from a large spectrum of affordability and punctuality requirements to 
achieve an end that is most desirable to them. 

Our market model is particularly simple, requiring only one parameter per flight from the airline company. 
Furthermore, we show that it admits a combinatorial, strongly polynomial algorithm for computing an equilibrium 
landing schedule and prices. 
\end{abstract}

\newpage

\section{Introduction}
Air Traffic Flow Management (ATFM) remains an unsatisfactorily solved problem despite massive efforts on the part of
the U.S. Federal Aviation Administration (FAA), airline companies, and even academia. 
Moreover, this comes despite the fact that, especially in inclement weather, this state of affairs 
results in huge monetary losses and delays\footnote{According to \cite{Bertsimas2}, the U.S. Congress Joint Economic Committee estimated that
in 2007, the loss to the U.S. economy was \$25.7 billion, due to 2.75 million hours of flight delays. In contrast, the total profit of 
U.S. airlines in that year was \$5 billion.}.
In a nutshell, the reason for this is that any viable solution needs to
satisfy two somewhat conflicting requirements: in addition to ensuring efficiency it also needs to be viewed
as ``fair'' by all parties involved. Indeed, \cite{Bertsimas3} state that `` ... While this work points at the possibility of dramatically
reducing delay costs to the airline industry vis-a-vis current practice, the vast majority of these proposals remain unimplemented. The
ostensible reason for this is fairness ... .''

Within academia, research on this problem started with the pioneering work of Odoni \cite{Odoni} and it flourished with the
extensive work of Bertsimas; we refer the reader to \cite{Bertsimas2, Bertsimas1} for thorough literature overviews and
an up-to-date list of references to important papers.

In this paper, we are proposing a radically different solution to the problem. Our solution is based on over a century
of work in mathematical economics that has established the free market principle of Adam Smith as a
remarkably efficient method for allocating scarce resources among alternative uses (sometimes stated in the colorful language
of the ``invisible hand of the market''). The goods in our novel market are delays and the buyers are airline companies; the latter
pay money to the FAA to buy away the desired amount of delay on a per flight basis. Our solution achieves efficiency as well as
fairness in a totally different manner than the solutions proposed so far.

In the centralized solutions proposed so far, FAA decides a schedule that is efficient, e.g., it decides
which flights most critically need to be served first in order to minimize cascading delays in the entire system.
Our contention is that in fact {\em the airlines themselves are the best judge of how to achieve efficiency.} They know best if a certain
flight needs to be served first because it is carrying CEOs of important companies who have paid a premium in order to reach their
destination on time or if delaying a certain flight by 30 minutes will not have dire consequences, however delaying it longer would 
propagate delays through their entire system and result in a huge loss. For such critical flights, our market-based solution gives the airline 
company the ability to pay, at the going rate, to buy away the desired amount of delay.

Interestingly, the sticky issue of fairness is simply finessed away by our solution -- whoever pays gets smaller delays. Moreover, it opens 
up the possibility of making diverse types of travelers happy through the following mechanism: the additional revenues generated by
FAA via our market gives it the ability to subsidize landing fees for low budget airlines. As a result, both types of travelers can
achieve an end that is most desirable to them, business travelers and casual/vacation travelers. The former, in inclement weather,
will not be made to suffer delays that ruin their important meetings and latter will get to fly for a lower price (and perhaps sip
coffee for an additional hour on the tarmac, in inclement weather, while thinking about their upcoming vacation).  

As is turns out, designing a market for this problem is not straightforward -- the requirements on the solution are many and
the classical market models of Fisher \cite{scarf} and Arrow-Debreu \cite{AD} simply do not apply to this setting. 
The most important requirements are:

\begin{itemize}
\item 
{\bf Simplicity:} The information needed from an airline for a particular flight should be very small (in particular, it should not be a utility
function, as needed in the classical market models). Typically flights have numerous special requirements, e.g., because of aircraft type.
Also, they have a myriad interdependencies with other flights -- because of the use of the same aircraft for subsequent flights,
passengers connecting with other flights, crew connecting with other flights, etc. Hence, the basic concept needs to be so simple that it 
leaves plenty of scope for dealing with these complexities.

\item
{\bf Efficient computability:} Today, the moderate and large airports handle hundreds of flights every day; in fact, the 30 busiest
airports handle anywhere from 1000 to 3000 flights per day. Hence, it is important that the schedule and prices be computable efficiently.
Moreover, to realize the desirable consequences of a competitive equilibrium, the prices should be such that there is parity between 
supply and demand -- and this despite recent negative results on computability of equilibria in even fairly special cases of 
traditional market models, e.g., see \cite{VY.plc}.
\end{itemize}

In our market model, the basic solution requires airlines to specify only one parameter, called criticality factor, for each of flight. 
A more elaborate solution, which allows airlines to present a much more complex set of requirements, requires them to specify a
delay function as well (see Section \ref{sec.salient}). We give a special LP in which parameters can be set according to the prevailing 
conditions on a particular day. It turns out that the underlying matrix of this LP is totally unimodular and hence it admits an integral 
optimal solution; such a solution yields an equilibrium schedule for the day. The dual of this LP yields equilibrium landing prices.
We further show that both can be found via an algorithm for the minimum weight perfect $b$-matching problem
and hence can be computed combinatorially in strongly polynomial time. 

In this paper, we handle the case of one airport only and we have deliberately kept the setting as simple
as possible so as to convey the main idea clearly, e.g., we have assumed that all flights are equivalent as far as assignment of
landing slots goes (this is of course far from true in practice, because of the widely varying sizes of aircraft).  
Considering the simplicity and computational efficiency of our solution, we believe there is plenty of potential for
extending this approach to more complex and realistic settings and even to generalize it to scheduling simultaneously for multiple
airports\footnote{According to \cite{Bertsimas2}, current research has mostly remained at the level of a single airport because 
of computational tractability reasons.}. We leave these exciting problems for future research.

\section{The Market Model}
\label{sec.model}

We will consider the problem of scheduling landings at one airport only.
Let $A$ be the set of all flights, operated by various airlines, that land in this airport in a day. 
For $i \in A$, the airline of this flight decides its {\em criticality factor}, denoted by $\al_i$. 
In a formal sense, $\al_i$ is the dollar value, as viewed by the airline, for a delay of unit time incurred by flight $i$. 
Thus, if this flight incurs a delay of $d$ units, the {\em cost of the delay} incurred by the airline is $\al_i \cdot d$ 
(for simplicity, we will first assume that the dollar value to the airline is a linear function of the delay; we will elaborate on this later).

We will assume that the entire day is partitioned into a set of landing time-slots in a manner that is most convenient for this airport.  
Let $B$ denote this set. Each slot $s$ has a capacity $c(s) \in \Zplus$ specifying the number of flights that can land in this time-slot. We will
assume that $c(s)$ is adjusted according to the prevailing weather conditions. 

After taking into consideration the prevailing weather conditions, for each flight $i$, FAA decides a {\em landing window}, denoted by $W(i)$,
which gives the set of time-slots available for this flight. Thus, if there are no delays, the earliest time-slot in $W(i)$
will be the scheduled arrival time of this flight (or a slightly earlier time, if there is a need to plan for before-time arrival).
If there are delays, $W(i)$ is adjusted by the FAA. The exact manner in which this adjustment is made affects the delays of flights
critically and therefore the FAA needs to have a well-thought-out policy for this. This adjustment also gives FAA a way of controlling
delays of individual flights; we explain this further in Section \ref{sec.salient}. For $s \in W(i)$, $d_i(s)$ is the delay incurred by 
flight $i$, beyond its scheduled arrival time, if it were to land in time-slot $s$.

A {\em landing schedule} is an assignment of flights to time-slots, respecting capacity constraints. Each time-slot will be assigned a
{\em landing price} which is the amount charged by FAA from the airline company if its flight lands in this time-slot. 
We will define the {\em total cost} incurred by a flight to be the sum of the price paid for landing and the cost of the delay.

We will say that a given schedule and prices are {\em an equilibrium landing schedule and prices} if:
\begin{enumerate}
\item
W.r.t. these prices, each flight incurs a minimum total cost.
\item
The landing price of any time-slot that is not filled to capacity is zero. This is a standard condition in economics; the
price of a good that is not fully sold must be zero.
\end{enumerate}

\section{LP Formulation and Algorithms}
\label{sec.LP}

In this section, we will give an LP that yields an equilibrium schedule; its dual will yield equilibrium landing prices. 
Section \ref{sec.rounding} shows how they can be computed in strongly polynomial time.

For $s \in W(i)$, $x_{is}$ will be the indicator variable that indicates whether flight $i$ is scheduled in time-slot $s$; naturally, in 
the LP formulation, this variable will be allowed to take fractional values. The LP given below obtains a fractional
scheduling that minimizes the total dollar value of the delays incurred by all flights, subject to capacity constraints of the time-slots.

\begin{lp}
\label{LP}
\minimize & \sum_{i \in A} \al_i \sum_{s \in W(i)} d_i(s) x_{is}  \\[\lpskip]
\st       & \forall i \in A: \ \sum_{s \in R(i)}  x_{is} \geq 1  \nonumber \\
          & \forall s \in B:  \ \sum_{i \ s.t. \ s \in W(i)}   x_{is}  \leq c(s)    \nonumber \\
          & \forall i \in A, \  s \in W(i):  \  x_{is}  \geq 0    \nonumber 
\end{lp}

Let $p_s$ denote the dual variable corresponding to the second set of inequalities. We will interpret $p_s$ as the
price of landing in time-slot $s$.  Thus if flight $i$ lands in time-slot $s$,
the total cost incurred by it is $p_s + \al_i \cdot d_i(s)$.
Let $t_i$ denote the dual variable corresponding to the first set of inequalities. In Lemma \ref{lem.cost} we will prove that $t_i$
is the total cost incurred by flight $i$ w.r.t. the prices found by the dual; moreover, each flight incurs minimum total cost.
  
The dual LP is the following.

\begin{lp}
\label{dual}
\maximize & \sum_{i \in A} t_i \  -  \ \sum_{s \in B} c(s) \cdot p_{s}  \\[\lpskip]
\st       & \forall i \in A, \ \forall s \in W(i):  \  t_{i}  \leq   p_s \  +   \  \al_i  \cdot d_i(s)      \nonumber \\
          & \forall i \in A: \  t_{i} \geq 0  \nonumber \\
          & \forall s \in B:  \   p_{s}  \geq  0    \nonumber 
\end{lp}

\begin{lemma}
\label{lem.cost}
W.r.t. the prices found by the dual LP (\ref{dual}), each flight $i$ incurs minimum total cost and it is given by $t_i$.
\end{lemma}

\begin{proof}
Applying complementary slackness conditions to the primal variables we get
\[ \forall i \in A, \ \forall s \in W(i):  \ \ x_{is} > 0 \ \  \implies \ \
t_{i}  =  p_s \  +   \  \al_i  \cdot d_i(s)  . \]
Moreover, for time-slots $s \in W(i)$ which are not used by flight $i$, i.e., for which $x_{is} = 0$, by the dual
constraint, the total cost of using this slot can only be higher than $t_i$. The lemma follows.
\end{proof}

The second condition required for equilibrium is satisfied because of complementarity applied to the variables $p_s$:
\[ \mbox{If} \ \ \sum_{i \ s.t. \ s \in W(i)}   x_{is}  < c(s), \ \ \mbox{then} \ \ p_s = 0 .\]

\subsection{Finding equilibrium prices and schedule in strongly polynomial time}
\label{sec.rounding}

Since the matrix underlying LP (\ref{LP}) is totally unimodular, it has an integral optimal solution.
In this section, we show that the problem of obtaining such a schedule is essentially a minimum
weight perfect $b$-matching problem and hence can be found in strongly polynomial time, see \cite{schrijver.book} Volume A. 
Furthermore, the dual variables on vertices, $p_s$s and $t_i$s computed by such an algorithm will yield equilibrium prices. 

Consider the edge-weighted bipartite graph $(A', B, E)$, with bipartition $A' = A \cup \{v\}$, where $A$ is the set of 
flights and $v$ is a special vertex, and $B$ is the set of time-slots. The set of edges $E$ and weights are as follows:
for $i \in A, \ s \in W(i)$, $(i, s)$ is an edge with weight $\al_i \cdot d_i(s)$, and for each $s \in B, \ (v, s)$ is an edge
with unit weight.

The matching requirements are: 1 for each $i \in A$, $c(s)$ for $s \in B$, and $\sum_{s \in B} {c(s)} - |A|$ for $v$;
clearly, we may assume that the last quantity is non-negative, because otherwise the LP is infeasible.

It is straightforward to verify that a minimum weight perfect $b$-matching in $G$, if it exists, gives an optimal solution to 
LP ($\ref{LP}$). If such a matching does not exist, the LP is infeasible. If so, the current capacity of time-slots is insufficient 
and either some landing windows need to be stretched appropriately and/or some flights need to be canceled
in order to obtain a schedule.

Hence we get.

\begin{theorem}
\label{thm.main}
There is a combinatorial, strongly polynomial algorithm for computing an equilibrium landing schedule and prices for the 
market model given above.
\end{theorem}

\section{Salient Features of our Solution}
\label{sec.salient}

We discuss below several salient properties of our solution and the reason it confirms to the requirements of a competitive market.

{\bf How should an airline set the criticality factor for a given flight:}
If this market is run in an open, complete-information manner, it will soon become obvious to airline companies how to set the criticality
factor of each flight suitably to achieve the desired end. The ``going rate'' for buying away delay will essentially be common knowledge.

{\bf How are prices set by the dual:} By standard LP theory, we get that the price dual variables, $p_s$, will adjust
according to the demand of each time-slot, i.e., a time-slot that is demanded by a large number of flights having large $\al_i$s will have 
a higher price and vice versa. In particular, if a slot is not allocated to capacity, its price will be zero as shown above.
Hence, the price of a slot will accurately reflect its demand.

{\bf How does $\al_i$ influence the price paid and delay incurred:}  We first prove the following lemma.

\begin{lemma}
\label{lem.alpha}
The criticality factor, $\al_i$, gives an upper bound on the amount of money 
that will be charged per unit of delay prevented on flight $i$, in an average sense. 
\end{lemma}

\begin{proof}
Let $p_t$ and $p_{t+k}$ be the prices of landing in slots that result in a delay of $t$ time units and $t+k$ time units, 
respectively. By Lemma \ref{lem.cost}, the former will be preferred to the latter only if $p_t - p_{t+k} \leq \al_i \cdot k$,
and if so, the average amount charged for saving the $k$ units of delay is at most $\al_i$.
\end{proof}

By Lemma \ref{lem.alpha}, a flight with a large $\al_i$ will 
incur a smaller delay, at the cost of paying more, and vice versa, as was desired.

{\bf Fixing the landing window:} A critical decision on the part of FAA, which has a large bearing on running this market in a fair manner, 
is how to fix the landing window for a flight in various operating conditions. Let us first consider an undesirable way of fixing the window. 
Suppose flights $f_1$ and $f_2$ have scheduled arrival times of $t$ and $t+k$, respectively. Suppose FAA assigns identical landing
windows for both flights, starting at time $t' \geq t+k$. Since the price of landing in any given time-slot is the same for both flights
and they both get the same set of options, flight $f_1$ will be experiencing an ``unfair'' additional delay of $k$ time units.

We have already indicated in Section \ref{sec.model} how to fix the landing window in case of no delays. In the presence of
delays, a simple and fair solution is to simply slide the window out equally for each flight by the minimum delay being experienced
at its scheduled landing time and, if the capacities of time-slots are significantly reduced, then stretch the window appropriately 
(otherwise LP (\ref{LP}) may not even be feasible).

Observe that the landing window can be used in another way by FAA -- it gives FAA a way of exercising control, e.g., if it needs to promote/demote
flights without changing the parameters provided by the airlines. 

{\bf Resorting to a non-linear function $d_i(s)$ :} We have stated above that $d_i(s)$ is simply the delay experienced by flight $i$,
beyond its scheduled landing time, if it were to land in time-slot $s$, i.e., it is a linear function of delay. Consider
the example given in the Introduction of a flight that will not have dire consequences if delayed
by 30 minutes, but will lead to cascading delays through the entire system if delayed longer. Interestingly, by moving to a more general 
setup in which airlines are allowed to provide a non-linear function $d_i(s)$ to the FAA, we can allow them to express their desired landing
times even in such complex situations. 

In the example stated above, this function would increase linearly in the first 30 minutes and
then increase very rapidly. As a consequence, the cost of the delay, and hence $t_i$, will increase very rapidly if the delay 
exceeds 30 minutes. Now in order to optimize its objective function, the LP would be willing to assign the earliest available  
landing spot, after 30 minutes, even if it is expensive. Observe that Lemma \ref{lem.cost} holds for an arbitrary function $d_i(s)$ as well.
Effectively, by increasing $d_i(s)$ rapidly beyond 30 minutes, the airline company is indicating that if the delay exceeds 30 minutes,
it is willing to go with a much more expensive landing spot in order to avoid further delays.

Observe that only monotonically increasing functions $d_i(s)$ make sense from the viewpoint of the airlines, since if this
function decreases, then flight $i$ may end up getting a slot with a high price and a large delay.

{\bf How fair is this solution:}  As indicated in the Introduction, our solution has the potential of enabling travelers from a large
spectrum of affordability and punctuality requirements to achieve an end that is most desirable to them. As long as the market is run
in a fair, open manner, and the airlines make clear to the potential passengers the trade offs, the choices made and what to expect, 
the issue of fairness becomes as moot as when a passenger decides whether to pay more and travel comfortably or pay less and risk
ruining their hurting back.

{\bf Incentive compatibility:} We claim, informally, that the airlines will not be able to game this solution by cleverly setting the
criticality factor: If they set $\al_i$ too high, they may end up paying for an expensive landing slot, and if they set it too low,
they may have to suffer a longer delay than desired. We leave the open problems of proving incentive compatibility in a formal sense
and comparing our solution to the VCG solution.

\section{Acknowledgments}

This idea was conceived in a lecture given by Dimitris Bertsimas in September 2010 at Georgia Tech. I am indebted to Dimitris for his fine
lecture and valuable discussions on the ATFM problem. I also wish to thank Kamal Jain for informing me of his paper \cite{tax} which led me
to the LP given here, and to Gagan Goel, Nimrod Megiddo, Laci Vegh and Mihalis Yannakakis for valuable feedback on a draft of this paper.

\bibliography{kelly}
\bibliographystyle{alpha}

\end{document}